%% file: main.tex
\documentclass[conference]{IEEEtran}
\IEEEoverridecommandlockouts

\usepackage{cite}
\usepackage{amsmath,amssymb,amsfonts}
\usepackage{algorithmic}
\usepackage{graphicx}
\usepackage{textcomp}
\usepackage{pstricks}
\usepackage{xcolor}
\usepackage{extarrows}

\input{macros}

\newtheorem{theorem}{Theorem}
\newtheorem{lemma}{Lemma}

\def\BibTeX{{\rm B\kern-.05em{\sc i\kern-.025em b}\kern-.08em
    T\kern-.1667em\lower.7ex\hbox{E}\kern-.125emX}}
\begin{document}

\title{On Secure Gradient Coding with Uncoded Groupwise Keys}

\author{
    \IEEEauthorblockN{Xudong You$^{1}$,  Kai Wan$^{1}$, Xiang Zhang$^{2}$, Wenbo Huang$^{1}$,  Robert Caiming Qiu$^{1}$, and Giuseppe Caire$^{2}$}
	\IEEEauthorblockA{  $^1$Huazhong University of Science and Technology, 430074 Wuhan, China \\ 
        $^2$Technische Universitat Berlin, 10587 Berlin, Germany \\ 
		Emails: \{xudong\_you, kai\_wan, eric\_huang, caiming\}@hust.edu.cn, \{xiang.zhang, caire\}@tu-berlin.de}
}

\maketitle

\begin{abstract}
This paper considers a new secure gradient coding problem with uncoded groupwise keys, formalized as a ${\sf (K, N, N_r, M, S)}$ secure gradient coding model, where a user aims to compute the sum of the gradients from ${\sf K}$ datasets with the assistance of ${\sf N}$ distributed servers. We consider  arbitrary heterogeneous data assignment, where each dataset is assigned to at least ${\sf M}$ servers. The user should recover the sum of gradients from the transmissions of any ${\sf N_r}$ servers. The security constraint guarantees that even if the user receives the transmitted messages from all servers, it cannot obtain any other information about the datasets except the sum of  gradients. Compared to existing secure gradient coding works, we introduce a practical constraint on secret keys, namely uncoded groupwise keys,  where the keys are mutually independent and each key is shared by precisely $\Ssf$ servers. An achievable secure gradient coding scheme with uncoded groupwise keys is proposed, which is then proven to be optimal if ${\sf S > M}$ and to be order optimal within a factor of $2$ otherwise. 
\end{abstract}


\section{Introduction}
\label{section: intro}

Distributed machine learning has become a crucial solution for training complex models on large datasets~\cite{dean2012large, ahmed2013distributed, li2014communication}. A primary challenge  is the high communication cost associated with transferring large amounts of data, such as gradient vectors~\cite{li2014communication}. Another key issue is dealing with stragglers: slow servers that can severely impact computational speed~\cite{dean2013tail}. Recently, coding techniques have been shown to effectively address both stragglers and adversarial server nodes, while also reducing communication overhead in distributed computations~\cite{li2015coded,speedup2018Lee, yu2019lagrange, yu2017polynomial, yu2020straggler, dutta2019optimal, dutta2016short}. 

Specifically, in~\cite{tandon2017gradient}, a framework called gradient coding was introduced, enabling the master node to collect and aggregate gradient vectors with minimal communication cost, while mitigating the impact of stragglers. In this setting, homogeneous data assignment was considered, where each server is assigned  the same number of datasets. Later in~\cite{ye2018communication}, the optimal trade-off  between communication cost, computation load (i.e., the number of datasets assigned to each server), and straggler tolerance was explored. Furthermore, the authors in~\cite{jahani2021optimal} considered arbitrary heterogeneous data assignment, where each dataset may be assigned to a different number of servers, and characterized the optimal communication cost. 

The secure aggregation model for gradient coding was originally proposed in~\cite{wan2022secure}. Besides the decodability constraint at the user side, an additional security constraint was added into the gradient coding problem, which imposes that the user cannot  recover  any other information about the datasets except the sum of gradients. For security, the servers should pre-store some keys assigned by a key distribution server.  
An information-theoretic lower bound on the total key size and several secure distributed coded computation schemes for specific data distribution patterns (e.g., cyclic, repetitive, and mixed) were proposed in~\cite{wan2022secure}, while achieving the optimal communication cost.

However, existing secure gradient coding schemes~\cite{wan2022secure} rely on coded keys, which typically require a trusted key server or private links among servers, limiting their practicality in decentralized systems.
Recently, uncoded groupwise keys were considered in the secure aggregation problem for federated learning (where each user holds an individual dataset)~\cite{wan2024uncoded}.\footnote{\label{foot:FL}An information-theoretic secure aggregation problem using uncoded groupwise keys was formulated in~\cite{wan2024uncoded}, following the secure aggregation problem for federated learning~\cite{bonawitz2017practical}. In~\cite{wan2024uncoded}, each distributed computing node holds its own local data to compute the gradient; thus, secure aggregation against stragglers (or user dropouts) requires two-round transmissions. In contrast, in the gradient coding problem, an additional data assignment phase for servers exists, and because of the redundancy in data assignment, one-round transmission is sufficient for secure aggregation.}
`Uncoded' means that  the keys are mutually independent, and `groupwise' means that each key is shared by a subset of servers. 
Such keys can be generated and shared directly among servers using standard key agreement protocols~\cite{diffie2022democratizing, maurer1993secret, ahlswede1993common, csiszar2004secrecy, gohari2010information, sun2022secure, sun2022compound, li2025capacity}, without requiring a trusted key server or private communication links. These properties make uncoded groupwise keys important for practical distributed learning systems.


This paper is the first which considers the secure gradient coding problem with uncoded groupwise keys. The user aims to compute the sum of gradients over $\Ksf$  datasets with the assistance of $\Nsf$ servers. After receiving the transmission from any ${\sf N_r}$ servers, the user should recover the sum of gradients without obtaining any other information about the datasets. 
By introducing a system parameter $\Ssf$, meaning that each uncoded key is shared by exactly $\Ssf$ servers, we aim to characterize the minimal (normalized) communication cost for arbitrary heterogeneous data assignment, where each dataset is assigned to at least ${\sf M}$ servers.

\subsection{Main Contributions}
Besides formulating the $(\Ksf, \Nsf, \Nsf_r, \Msf, \Ssf)$ secure gradient coding problem with uncoded groupwise keys and  arbitrary heterogeneous data assignment, 
we propose a universal secure gradient coding scheme with an explicit achievable communication cost. 
Our scheme jointly exploits data redundancy and key redundancy to achieve the minimum communication cost required for secure gradient aggregation.
Compared to the communication cost of  the optimal non-secure gradient coding scheme, the proposed scheme achieves the same communication cost when ${\sf S > M}$, and achieves a communication cost within a multiplicative gap of factor $2$ when ${\sf S \leq M}$.

\textit{Paper Organization.}
This paper is organized as follows. Section~\ref{sec: system} introduces the secure gradient coding model with uncoded groupwise keys. Section~\ref{sec: main} presents the main results. Section~\ref{sec: general scheme} describes the proposed general coding scheme. Finally, section~\ref{sec: conclusion} concludes this paper.

\textit{Notation.}
Calligraphic symbols denote sets, bold symbols denote vectors and matrices, and sans-serif symbols denote system parameters. We use $\mid \cdot \mid $ to represent the cardinality of a set or the length  of a vector; $[a : b] := \{a,a+1,\dots,b\}$ and $[n] := [1 : n]$; $\mathbf{1}_n$ and $\mathbf{0}_n$ represent the vertical $n$-dimensional vector whose elements are all 1 and all 0, respectively; $\mathbf{A}^\Tsf$ and $\mathbf{A}^{-1}$ represent the transpose and the inverse of matrix $\mathbf{A}$, respectively; 
the matrix $[\mathbf{A};\mathbf{B}]$ is expressed in a format similar to Matlab, equivalent to $\begin{bmatrix}
     \mathbf{A}\\ 
     \mathbf{B}
 \end{bmatrix} $;$
\mathbf{A}(\Nc,\cdot)$ and $\mathbf{A}(\cdot,\Oc)$ represent the sub-matrices of matrix $\mathbf{A}$ whose rows and columns in the set $\Nc$, respectively; $\mathbf{A}(\Nc,\Oc)$ represents the sub-matrix of matrix $\mathbf{A}$ whose rows are in the set $\Nc$ and columns are in the set $\Oc$; $\mathbf{I}_n$ represents the identity matrix of dimension $n \times n$; $\mathbf{0}_{m \times n}$ represents all-zero matrix of dimension $m \times n$; $\mathbf{A}_{m \times n}$ explicitly indicates that the matrix $\mathbf{A}$ is of dimension $m \times n$; let $\binom{x}{y} = 0$ if $x < 0$ or $y < 0$ or  $x < y$; let $\binom{\Xc}{y} = \{\Sc \subseteq \Xc:|\Sc|=y\}$, where $|\Xc|\ge y>0$; For a set $\Sc$, we denote the $i^{\text{th}}$ smallest element by $\Sc(i)$; $\mathbb{F}_{\sf q}$ represents a finite field with order ${\sf q}$. In the rest of the paper, entropies will be in base ${\sf q}$, where ${\sf q}$ represents the field size.

\section{System Model}
\label{sec: system}
Given a training dataset $D=\{(x_i,y_i)\}_{i=1}^{d}$ and let $\boldsymbol{\beta}$ be the model parameters. Given a loss function $\mathcal{L}(\cdot)$, the goal of training is to minimize:
\begin{equation}
\mathcal{L}(D;\boldsymbol{\beta})
= \sum_{i=1}^{d} \mathcal{L}(x_i,y_i;\boldsymbol{\beta}).
\end{equation}

Gradient-based methods update the model parameters using the full gradient $\mathbf{g} = \nabla_{\boldsymbol{\beta}} \mathcal{L}(D;\boldsymbol{\beta})$, which is repeatedly evaluated during training. To enable distributed computation, the dataset $D$ is partitioned into $\Ksf$ disjoint subsets $\{D_1,\ldots,D_{\Ksf}\}$. Then the main task is to compute the sum of gradients $\mathbf{g}=\sum_{k=1}^{\Ksf}\mathbf{g}_k$, where $g_k\in \mathbb{F}_q^\Lsf$ is the  gradient over $D_k$ for $k \in [\Ksf]$. Assume that the gradients are mutually independent.

Then we formulate the $(\sf K, N, N_r, M, S)$ secure gradient coding problem with uncoded groupwise keys and arbitrary heterogeneous data assignment. A user aims to compute the sum of gradients over $\Ksf$ datasets $\{D_1,\ldots,D_{\Ksf}\}$ with the assistance of any ${\sf N_r}$ servers selected from a total of $\Nsf$ servers, while learning no additional information about the individual gradients. 
For each subset $\Vc \subseteq [\Nsf]$ with $|\Vc|=\Ssf$, a random key $Q_{\Vc}$ is generated and shared exclusively among the servers in $\Vc$ using key agreement protocols. These keys are referred to as uncoded groupwise keys. All keys are mutually independent and are also independent of the gradients, 
i.e., 
\begin{equation}
\begin{split}
    &H\left(\left(Q_\Vc:\Vc \in \binom{[\Nsf]}{\Ssf}\right), \left(g_1,...,g_\Ksf\right)\right) \\
    =&\sum_{\Vc \in \binom{[\Nsf]}{\Ssf}}H(Q_\Vc)+\sum_{k \in [\Ksf]}H(g_k).
    \label{con: independent constraints}
\end{split}
\end{equation}

A secure gradient coding scheme consists of three phases: data assignment, encoding, and decoding.

\textit{Data Assignment}:
Each dataset $D_k$ is assigned to a subset of servers $\Dc_k \subseteq [\Nsf]$ with $|\Dc_k| \ge \Msf$, where $\Msf$ denotes the minimum replication factor. For each server $n \in [\Nsf]$,   $\Zc_n \subseteq [\Ksf]$ denotes the set of datasets assigned to it. Assume that $\Zc_1,\ldots,\Zc_{\Nsf}$ are pre-fixed; that is, we consider arbitrary heterogeneous data assignment.

\textit{Encoding}:
Each server $n\in [\Nsf]$ computes the gradients $g_k$ locally for all $k\in \Zc_n$. 
Based on these gradients and the assigned  keys, server $n$ generates a coded message $X_n$, which is then transmitted to the user, i.e., 
\begin{equation}
    H \left( X_n|\left(g_k: k \in \Zc_{n}\right), \left(Q_{\Vc}: n\in \Vc\right) \right) = 0, \forall n \in [{\sf N}].
    \label{con: data and key constraints}
\end{equation}

\textit{Decoding}:
The user collects coded messages from any subset of servers $\Uc \subseteq [\Nsf]$ with $|\Uc|= {\sf N_r}$. Based on the received messages $\mathcal{X} = \{X_n: n \in \Uc\}$, the user must be able to recover the sum of gradients:
\begin{equation}
    H\left(\sum_{k\in [\Ksf]}g_k|\left(X_n:n\in \Uc\right)\right)=0, \forall \Uc \in \binom{\sf [N]}{\sf N_r}.
\end{equation}

To enable cancellation of all keys through linear operations, each key must appear in at least two received messages. Otherwise, a key that appears only once cannot be eliminated. This requirement leads to the following feasibility condition:
\begin{equation}
\label{eq: group size constraint}
    \sf S \ge N-N_r+2.
\end{equation}

In addition to decodability, the scheme must satisfy an information-theoretic security requirement. Specifically, even if the user received transmissions from all servers, it should not be able to obtain any information about the individual gradients beyond their sum. This requirement is formalized as:
\begin{equation}
    I\left(X_{1},\ldots,X_{\sf N};g_{1},\ldots, g_{\sf K} | \sum_{k\in [\Ksf]}g_k \right) = 0.
    \label{con: security constraint}
\end{equation}

  The communication cost is defined as the maximum normalized message size transmitted by any server: 
\begin{equation}
{\sf R} \overset{\Delta}{=} \max_{n \in [\Nsf]} \frac{|X_n|}{\Lsf}.
\end{equation}
The objective is to characterize the optimal communication cost for arbitrary heterogeneous data assignment $\Rsf^*$.

The optimal communication cost of non-secure gradient coding under linear coding was characterized in~\cite{jahani2021optimal}.
\begin{lemma}[\cite{jahani2021optimal}]
\label{lm: optimal communication cost}
For the ${\sf (K, N, N_r, M)}$ non-secure gradient coding model with arbitrary data assignment, where each dataset is assigned to at least $\Msf$ servers, the optimal communication cost under linear coding is equal to
\begin{equation}
\Rsf_{\rm n}^*=\frac{1}{\sf N_r - N + M}.
\end{equation}  
\end{lemma}
Clearly, we have $\Rsf^* \geq \Rsf_{\rm n}^*$. 
It was shown in~\cite{wan2022secure} that without constraints on shared keys, security incurs no additional communication cost, whereas here we impose uncoded groupwise key constraints. Moreover, our scheme reduces to the model in~\cite{wan2022secure} when $\Ssf = \Nsf$, and recovers the cyclic results.

\section{Main Results}
\label{sec: main}
For the $(\sf K, N, N_r, M, S)$ secure gradient coding problem with uncoded groupwise keys and arbitrary heterogeneous data assignment, this paper proposes a universal secure
gradient coding scheme as follows, whose proof could be found in Section~\ref{sec: general scheme}.
\begin{theorem}[Achievable Communication Cost]
\label{thm: communication cost}
For the $(\sf K, N, N_r, M, S)$ secure gradient coding problem with uncoded groupwise keys and arbitrary heterogeneous data assignment, the following communication cost is achievable:
\begin{align}
\label{eq: communication cost}
    \sf R=\frac{\binom{N}{S}-\binom{M}{S}}{\left(\binom{N}{S}-\binom{M}{S}\right)N_r-\binom{N}{S}(N-M)}.
\end{align}
\end{theorem}

The communication cost is determined by how many groupwise keys cannot be canceled using data redundancy. When fewer servers share a key, additional communication is required to eliminate it. Larger key groups increase overlap and reduce this overhead.

We next provide the (order) optimality results of the proposed scheme, whose proof could be found in Appendix~\ref{pr: order optimal}.
\begin{theorem}[(Order) optimality)]
When $\sf S>M$, the  achievable communication cost coincides with the optimal communication cost of   non-secure gradient coding under linear coding,  
\begin{equation}
    \sf \Rsf^* =\Rsf_{\rm n}^*=\frac{1}{N_r-N+M};
\end{equation}
when ${\sf N-N_r+2} \le \sf S \leq M$, the achievable communication cost is order-optimal within a factor of $2$,
\begin{equation}
    \Rsf_{\rm n}^* \le \Rsf \leq 2\Rsf_{\rm n}^*.
\end{equation}
\end{theorem}
Note that when $\rm S \le M$, limited key sharing prevents full cancellation through data redundancy, leading to extra communication. However, this overhead is bounded, and the communication cost remains within a constant factor of the optimum.
It then reveals a sharp transition at $\sf S = M$: above this threshold, security is essentially free, while below it, security incurs an order-optimal communication overhead.

\paragraph*{Numerical evaluation}
We then compare our communication cost with the optimal cost of non-secure gradient coding in~\cite{jahani2021optimal}. 
Fig.~\ref{fig: Numerical analysis M} shows that the communication cost decreases as $\Msf$ increases. When $\Ssf > \Msf$, our scheme matches the optimal cost, confirming that sufficient key overlap eliminates extra overhead.
Fig.~\ref{fig: Numerical analysis S} illustrates the effect of $\Ssf$. With small $\Ssf$, the cost is higher due to limited key sharing, but it quickly drops as $\Ssf$ grows. When $\Ssf \geq \Msf$, the cost achieves the optimum.
Overall, the numerical results demonstrate that our scheme remains efficient across parameter settings and achieves near-optimal performance. 

\begin{figure}
  \centering
  \includegraphics[width=0.31\textwidth]{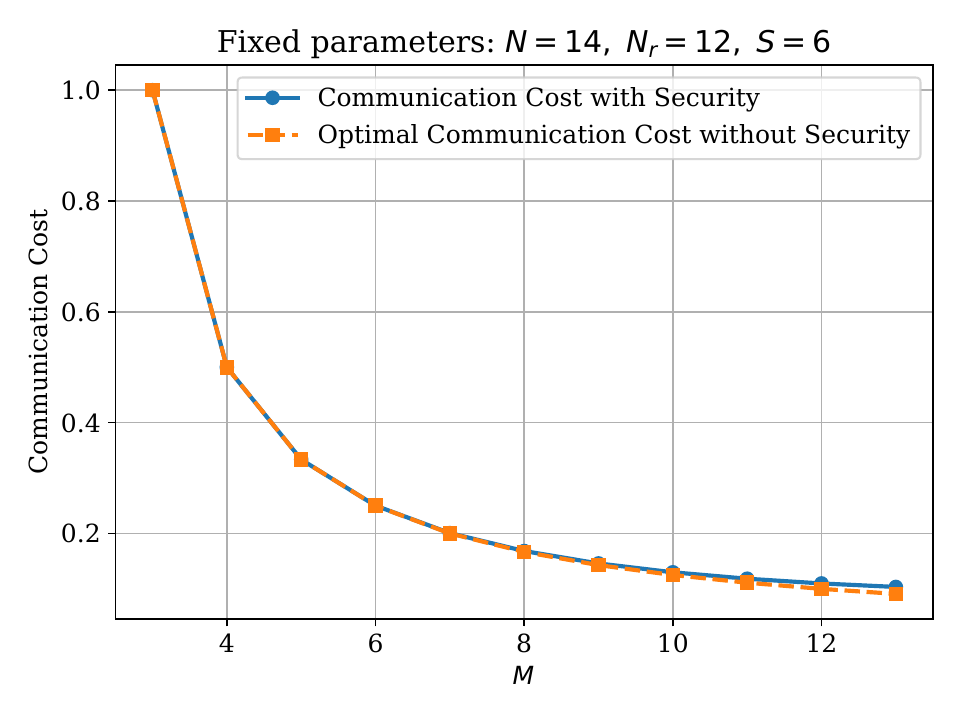}
  \vspace{-.4cm}
  \caption{$\Nsf=14, {\sf N_r}=12, \Ssf=6$, with varying $\Msf\in [3:13]$.}
  \label{fig: Numerical analysis M}
\end{figure}

\begin{figure}
  \centering
  \includegraphics[width=0.32\textwidth]{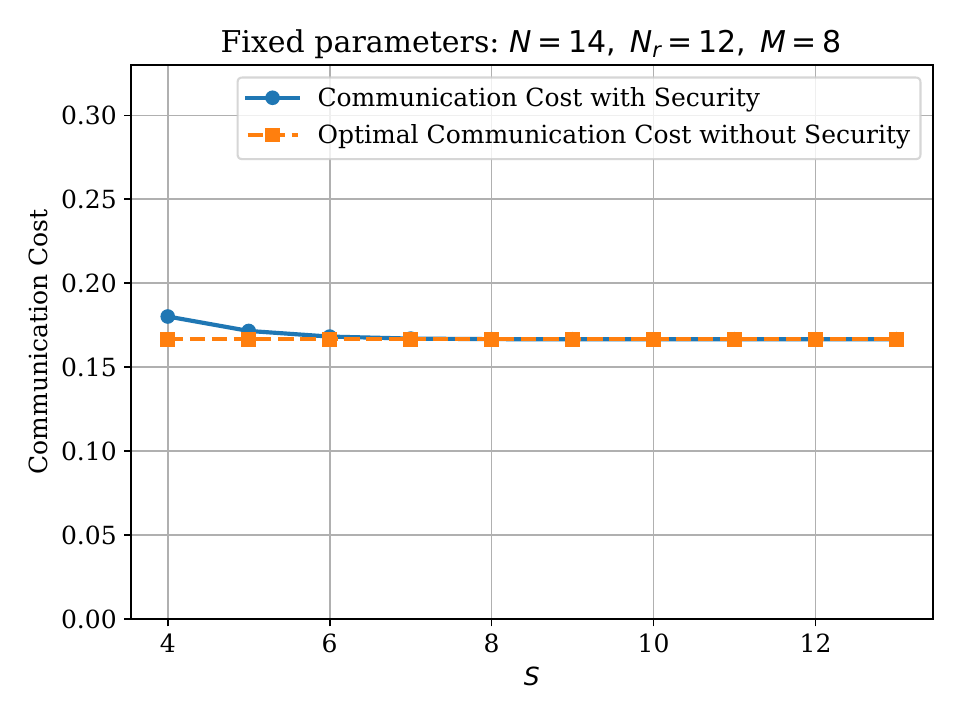}
  \vspace{-.4cm}
  \caption{$\Nsf=14, \Nsf_{\rm r}=12, \Msf=8$, with varying $\Ssf\in [4:13]$.}
  \label{fig: Numerical analysis S}
\end{figure}

\section{Proof of Theorem~\ref{thm: communication cost}: Achievable scheme}
\label{sec: general scheme}

\subsection{Example: ${\sf (K, N, N_r, M, S)} = (3, 3, 3, 2, 2)$}
\label{sec: example}
We first illustrate the proposed scheme through an example, where ${\sf (K, N, N_r, M, S)} = (3, 3, 3, 2, 2)$. 

\paragraph*{Data assignment}
Each dataset $D_i$ is assigned to $2$ distinct servers. Specifically, in this example, $\Dc_1=\{2,3\}, \Dc_2=\{1,2\}, \Dc_3=\{1,2\}$.
For each $\Vc\in \binom{[3]}{2}$, the servers in $\Vc$ share a key $Q_\Vc$ with $\frac{\Lsf}{3}$ uniformly i.i.d. symbols on $\mathbb{F}_q$. In this example, for the ease of illustration, we   assume that $q$ is a large prime number, which  is not required in our general scheme.
The assignment of datasets and keys is shown as Fig.~\ref{fig: example1 data assignment}. 
\begin{figure} 
  \centering  
  \includegraphics[width=0.3\textwidth]{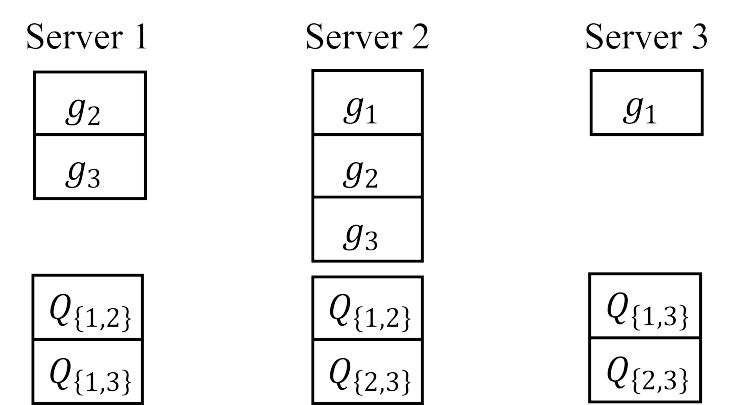}  
  \vspace{-.2cm}
  \caption{Data and key assignment for ${\sf (K, N, N_r, M, S)} = (3, 3, 3, 2, 2)$.}  
  \label{fig: example1 data assignment}  
\end{figure}  

\paragraph*{Encoding}
To achieve our communication cost ${\sf R}=2/3$ in Theorem~\ref{thm: communication cost}, each gradient $g_k$ is partitioned into $3$ non-overlapping and equal-length pieces $g_k = (g_{k,1}, g_{k,2}, g_{k,3})$, each piece containing $\Lsf/3$  symbols. Each server then transmits $2$ independent linear combinations of  pieces and keys to the user. After the user receives $6$ transmissions from three servers, the user can recover $\Fm\Wm$. More precisely,   $\Fm$ is a matrix with dimension $6\times12$,  
\begin{equation}  
  \Fm =  
  \begin{bmatrix}  
  \Fm_1 & {\bf 0}_{3 \times 3} \\  
  \Fm_2 & \Fm_3
  \end{bmatrix},  
\end{equation}
where $\Fm_1$ specifies the computation demand
  \begin{align}  
      \Fm_1 = \left[\begin{array}{ccccccccc}  
           111&000&000  \\  
           000&111&000  \\  
           000&000&111  
      \end{array}\right]_{3 \times 9},  
  \end{align}  
and we set $\Fm_3={\bf I}_3$ while $\Fm_2$ will be designed later. 
 $\Wm$ is  a matrix  composed of gradient pieces and keys given by
\begin{equation}  
    \Wm = [g_{1,1}; g_{2,1}; g_{3,1}; g_{1,2}; \dots; g_{3,3}; Q_{1,2}; Q_{1,3}; Q_{2,3}].  
\end{equation}  

Each server $n\in[3]$ transmits a coded message
$X_n = \Cm_n\Fm\Wm,$ 
where $\Cm_n$ is an encoding matrix with dimension $2\times6$ to be determined. 

\paragraph*{Decoding}
After receiving transmissions $\mathcal{X} = \{X_n: n\in[3]\}$ from servers in $\Uc=\{1,2,3\}$, the user can decode the received messages as:
\begin{equation}
  \Cm_\Uc^{-1}\mathcal{X} = 
  \begin{bmatrix}
  \Fm_1 & \mathbf{0} \\
  \Fm_2 & \mathbf{I}_3
  \end{bmatrix}
  \Wm,
\end{equation}
where $\Cm_\Uc=[\Cm_1;\Cm_2;\Cm_3]$ and the sum of gradients $g_1+g_2+g_3$ is recovered from the first three rows of $\Fm\Wm$.

To guarantee the successful transmission at each server, the correct recovery of $\Fm\Wm$, and the information-theoretic security of the individual gradients, we impose some design constraints on $(\Cm_n:n\in[3])$ and $\Fm_2$, which we describe next.

{\it \textbf{Encodability constraint}}. To ensure that each server  $n\in[\Nsf]$ only transmits the gradients and keys assigned to it, we  have
\begin{equation}
\label{eq: encoding constraint}
    \Cm_n\Fm(\cdot, \overline{\Cc_n})=\mathbf{0}_{2\times |\overline{\Cc_n}|},
\end{equation}
where $r$ denotes the number of linear combinations of gradient pieces and keys transmitted by each server (in  this example $r=2$), and $\overline{\Cc_n}$ denotes the set of columns in $\Fm$ corresponding to gradients (or gradient pieces) and keys not assigned to server $n$. This condition guarantees that in the message $X_n$ sent by server $n\in[\Nsf]$, the coefficients corresponding to the unassigned gradients  (or gradient pieces)  and keys are $0$.

{\it \textbf{Decodability constraint}}. The user should be able to recover $\Fm\Wm$ from the transmissions of any $\sf N_r$ servers. Accordingly, for any $\Uc\in\binom{\sf [N]}{\sf N_r}$, the matrix
\begin{equation}
\label{eq: decoding constraint}
    \Cm_\Uc = \begin{bmatrix}
        \Cm_{\Uc(1)} \\
        \vdots \\
        \Cm_{\Uc(\sf N_r)}
    \end{bmatrix}
\end{equation}
should be full rank.

Now it remains to design $\Cm$ and $\Fm_2$ satisfying the encodability  and decodability constraints. In the following, we will construct $\Cm$ to satisfy the encodability constraint of key part and decodability constraint, and satisfy the encodability constraint of gradient part by the design of $\Fm_2$.

\subsubsection{Design Matrix $\Cm$}
For server $1$, who has the keys $Q_{1,2}$ and $Q_{1,3}$  without knowing $Q_{2,3}$, its transmitted coefficients associated with unassigned keys $Q_{2,3}$ should be $0$, which imposes the following constraint:
\begin{equation}
    \Cm_1\Fm(\cdot,\{12\})= 
    \Cm_1
    \begin{bmatrix}
    0, 0, 0, 0, 0, 1  
    \end{bmatrix}^{\sf T}
    = \mathbf{0}_{2 \times 1}.
\end{equation}

This constraint can be satisfied by generating $\Cm_1$ from the left nullspace of $\Fm(\cdot,\{12\})$. Applying the same principle to the remaining servers, we complete the construction of $\Cm$, which is full rank and thus satisfies the decodability constraint. 
\begin{align*}
  \Cm=
  \begin{bmatrix}
      * & * & * & * & * & 0 \\
      * & * & * & * & * & 0 \\
      * & * & * & * & 0 & * \\
      * & * & * & * & 0 & * \\
      * & * & * & 0 & * & * \\
      * & * & * & 0 & * & * 
  \end{bmatrix}=
  \begin{bmatrix}
      1 & 2 & 3 & 1 & 0 & 0 \\
      2 & 3 & 4 & 0 & 1 & 0 \\
      3 & 2 & 5 & 1 & 0 & 0 \\
      8 & 5 & 7 & 0 & 0 & 1 \\
      9 & 7 & 2 & 0 & 1 & 0 \\
      7 & 4 & 1 & 0 & 0 & 1
  \end{bmatrix},
\end{align*}
where `*' represents an i.i.d. uniform element on $\mathbb{F}_{q}$. 

\subsubsection{Design Matrix $\Fm_2$}
For gradient $g_1$, which is stored only at servers $2$ and $3$, server $1$ must not transmit any piece of $g_1$. This requires
\begin{align}
    \Cm_1\Fm(\cdot,\Gc_1)=\begin{bmatrix}
  1 & 2 & 3 & 1 & 0 & 0 \\
  2 & 3 & 4 & 0 & 1 & 0 
  \end{bmatrix}
  \begin{bmatrix}
  1 & 0 & 0 \\
  0 & 1 & 0 \\
  0 & 0 & 1 \\
  * & * & * \\
  * & * & * \\
  * & * & *
  \end{bmatrix}
  = \mathbf{0}_{2 \times 3},
\end{align}
where $\Gc_1=\{1+3(i-1):i\in[3]\}$ denotes the set of columns in $\Fm$ corresponding to gradient $g_1$.  
Since $\Cm_1(\cdot,\{4,5,6\})$ is left full rank, this system is solvable. By applying the same procedure to the other gradients, we obtain the complete construction of $\Fm_2$.

At this point, both $\Cm$ and $\Fm_2$ have been fully specified and satisfy the constraints.

{\it Security:} 
We provide an intuitive explanation of this security constraint. The overall information obtained by the user, $\Fm\Wm$, can be horizontally partitioned into two blocks. The first block $[\Fm_1,\mathbf{0}]\Wm$ is  the desired sum of gradients $\Fm_1[g_{1,1};\ldots;g_{3,3}]$. Hence, the security is determined by the second block 
\begin{equation}
    [\Fm_2,\Fm_3]\Wm =
    \left[\begin{array}{cccccc}  
      * & \cdots & * & 1 & 0 & 0 \\  
      * & \cdots & * & 0 & 1 & 0 \\  
      * & \cdots & * & 0 & 0 & 1 
  \end{array}\right]
  \left[\begin{array}{c}  
      g_{1,1} \\  
      \vdots \\  
      g_{3,3} \\
      Q_{1,2} \\  
      Q_{1,3} \\  
      Q_{2,3} 
  \end{array}\right],
\end{equation}
where $\Fm_2$ multiplies individual gradients and must be protected, while $\Fm_3$ multiplies the keys. If $\Fm_3$ is full rank, no linear transformation can eliminate any of its rows. As a result, each row of $\Fm_2$ is always masked by at least one independent key. As $\Fm_3$ is designed as an identity matrix, the security constraint is satisfied.

\subsection{General scheme}
We are now ready to describe our general scheme for the ${\sf (K, N, N_r, M, S)}$ secure gradient coding model.  It has been proved in~\cite{zhao2022information} that, by assuming that $L$ is large
enough, 
 we can directly set that $q$ is large enough.
 Denote the communication cost in~\eqref{eq: communication cost} as $\Rsf=r/n$ for simplicity.

\emph{Data assignment:}
The datasets are assigned to servers in an arbitrary way. Each dataset $D_k$ is replicated at least ${\sf M}$ times among the servers in the set $\mathcal{D}_k$. The keys are generated in an uncoded groupwise manner. Each group $\mathcal{V}\in \sf \binom{[N]}{S}$ shares a unique key $Q_\mathcal{V}$. To simplify notation, we replace the key notation from set $Q_\mathcal{V}$ to index $Q_i$, where $\mathcal{V}$ is the $i$-th set in $\binom{[\Nsf]}{\Ssf}$. Each key has $\Lsf (\Nsf-\Msf)/n$ i.i.d. uniform symbols in $\mathbb{F}_{q}$. 

\emph{Encoding:}
To achieve our communication cost ${\sf R}=r/n$, each gradient $g_k$ is partitioned into $n$ non-overlapping and equal-length pieces $g_k = (g_{k,1},\ldots, g_{k,n})$. each $Q_i$ is further partitioned into $\Nsf-\Msf$ non-overlapping and equal-length pieces $Q_i = (Q_{i,1},\ldots, Q_{i,\Nsf-\Msf})$. So each gradient piece and  each key piece  both contain $\Lsf /n$ i.i.d. uniform symbols in $\mathbb{F}_{q}$.
Each server then transmits $r$ independent linear combinations of gradient and key pieces to the user. After the user receives $r{\sf N_r}$ transmissions from any ${\sf N_r}$ servers, the user can recover $\Fm\Wm$. The demand matrix $\Fm$ is defined as:
\begin{equation}  
  \Fm =  
  \begin{bmatrix}  
  \Fm_1 & {\bf 0}_{n \times \sf(N-M)\binom{N}{S}} \\  
  \Fm_2 & \Fm_3
  \end{bmatrix},  
\end{equation}
where $\Fm_1$ specifies the computation demand
\begin{equation}
    \begin{bmatrix}
        {\bf 1}_\Ksf & \cdots & {\bf 0}_\Ksf  \\
        \vdots & \ddots & \vdots \\
        {\bf 0}_\Ksf &\cdots & {\bf 1}_\Ksf
    \end{bmatrix}_{n\times n\Ksf},
\end{equation}
and we set $\Fm_3={\bf I}_{\sf(N-M)\binom{N}{S}}$ while $\Fm_2$ will be designed later.

The message matrix $\Wm$ is given by
\begin{equation}  
    \Wm = [g_{1,1}; g_{2,1};...;g_{\Ksf,n};Q_{1,1};Q_{2,1};...;Q_{\binom{\Nsf}{\Ssf},\Nsf-\Msf}].  
\end{equation}  

Each server $n\in[\Nsf]$ transmits a coded message
\begin{equation}
    X_n = \Cm_n\Fm\Wm,
\end{equation}
where $\Cm_n$ is an encoding matrix with dimension $r\times r{\sf N_r}$. 

\emph{Decoding:}
After receiving  $\mathcal{X} = \{X_n:n\in \Uc\}$ from servers in $\Uc\in\binom{[\Nsf]}{\sf N_r}$, the user can decode the received messages as:
\begin{equation}
  \Cm_\Uc^{-1}\mathcal{X} = 
  \begin{bmatrix}
  \Fm_1 & \mathbf{0} \\
  \Fm_2 & \mathbf{I}
  \end{bmatrix}
  \Wm,
\end{equation}
from which the sum of gradients $g_1+\ldots+g_\Ksf$ is recovered from the first $n$ rows of $\Fm\Wm$.

To ensure the correctness of the proposed scheme, the encodability and decodability constraints introduced in Section~\ref{sec: example} should be satisfied, by designing $\Cm$ and $\Fm_2$. 

For each server $n\in [\Nsf]$, let $\Sc_n \subseteq \binom{\Nsf}{\Ssf}$ denote the set of indices of keys available to server $n$, and let $\overline{\Sc_n}$ denote its complement. In the transmitted message $X_n$, the coefficients corresponding to unassigned keys $\overline{\Sc_n}$ are zero; thus
\begin{equation}
  \Cm_{n}\Fm(\cdot, \overline{\Hc_{n}}) = \mathbf{0},
\end{equation}
where $\overline{\Hc_{n}}=\{n\Ksf+i+(j-1)\binom{\Nsf}{\Ssf}: i \in \overline{\Sc_{n}}, j \in [{\sf N-M}]\}$ indexes the columns of $\Fm$ corresponding to keys unavailable to server $n$. 
We construct $\Cm_{n}$ by selecting its rows from the left nullspace of $\Fm(\cdot, \overline{\Hc_{n}})$. Applying this procedure to all servers yields the complete encoding matrix $\Cm$, which is full rank and thus satisfies the decodability constraint.\footnote{\label{foot: full rank}$\Cm_n$ is generated from the nullspace of $\Fm(\cdot, \overline{\Hc_{n}})$, whose available dimension is proportional to the keys available to server $n$. For any servers in $\Uc\in\binom{[\Nsf]}{\sf N_r}$, all keys are available, ensuring $\Cm_\Uc$ contains sufficient independent vectors. Consequently, $\Cm_\Uc$ is full rank with high probability. The detailed decodability proof is provided in Appendix~\ref{pr: decodability}.}

After determining $\Cm$, we determine $\Fm_2$. 
For each gradient $g_k$, which is known by the servers in $\Dc_k$, the  servers in $\overline{\Dc_k}$ cannot transmit the corresponding gradient pieces; thus
\begin{align}
    \Cm(\overline{\Dc_k},\cdot)\Fm(\cdot,\Gc_k) = \mathbf{0},
\end{align}
where $\Gc_k=\{k+(i-1)\Ksf: i\in[n]\}$ indexes the columns of $\Fm$ corresponding to $g_k$. 
These constraints reduce to affine linear systems in $\Fm_2$. Since $\Cm(\overline{\Dc_k},\Oc)$ is left full rank,\footnote{\label{foot: equation}This property is ensured by our communication cost design, which is derived from the requirement that $\Cm(\overline{\Dc_k},\Oc)$ is left full rank. The detailed proof is in Appendix~\ref{pr: encodability}.} where $\Oc$ is the set of columns in $\Cm$ multiplied by $\Fm_2$, these constraints are solvable. Solving these systems for all $k\in[\Ksf]$ completes the construction of $\Fm$.

At this point, both $\Cm$ and $\Fm_2$ have been fully specified and satisfy both constraints.

{\bf Proof of security.} We provide an intuitive proof on the security in the following, where the formal information theoretic proof is provided in Appendix~\ref{pr: security}. 
The overall information obtained by the user, $\Fm\Wm$, can be horizontally partitioned into two blocks as in the above example. The first block $[\Fm_1, \mathbf{0}]\Wm$ is the computation task. The security is mainly concentrated on second block $[\Fm_2,\Fm_3]\Wm$, where $\Fm_2$ contains the information of individual gradients and must be protected. Since $\Fm_3$ is an identity matrix, no linear transformation can eliminate any of its rows. Hence, each row of $\Fm_2$ is always protected by at least one independent key, which satisfies the security constraint.

{\bf Communication cost.}
\label{communication cost}
We derive the communication cost from the encodability constraint of the gradient part, which requires $\Cm(\overline{\Dc_k},\Oc)$ to be full-rank. 
Assume $\Rsf=\frac{r}{n}$ and each key is partitioned into $\alpha$ pieces. The dimension of the available subspace for $\Cm(\overline{\Dc_k},\Oc)$ is $\alpha\left({\sf \binom{N}{S}-\binom{M}{S}}\right)$,\footnote{The detailed derivation is provided in Appendix~\ref{pr: communication cost}.} 
which must be no smaller than the number of rows, yielding the constraint $\alpha\left({\sf \binom{N}{S}-\binom{M}{S}}\right) \ge (\Nsf-\Msf)r$. 
The coding matrix $\Cm$ has dimension $r\Nsf \times r{\sf N_r}$, and the demand matrix $\Fm$ has dimension $\left(n+\alpha\binom{\Nsf}{\Ssf}\right) \times \left(n\Ksf+\alpha\binom{\Nsf}{\Ssf}\right)$. For the product $\Cm\Fm$ to be well defined, we have $r{\sf N_r}=n+\alpha\binom{\Nsf}{\Ssf}$. Finally, eliminating $\alpha$, we obtain the communication cost in Theorem~\ref{thm: communication cost}.

\section{Conclusion}
\label{sec: conclusion}
In this paper, we proposed a novel secure gradient coding model with uncoded groupwise keys. Under arbitrary data assignment, we present a general communication cost together with an achievable scheme. The communication cost is shown  to be optimal when $S>M$ and order-optimal within a factor of 2 of the optimum when $S \le M$, where $S$ and $M$ denote the key group size and minimum data replication factor, respectively.

\section*{Acknowledgment}
The work of X. You, K. Wan was funded by NSFC-12141107 and Wuhan “Chen Guang” Pragram under Grant 2024040801020211. The work of X. Zhang and G. Caire was supported by the Gottfried Wilhelm Leibniz-Preis 2021 of the German Science Foundation (DFG).

\clearpage

\bibliographystyle{IEEEtran}
\bibliography{ref}
\clearpage

\appendices
\section{Order-Optimal Proof}
\label{pr: order optimal}
We give the order-optimal proof of our scheme.

When $\sf S>M$, we have ${\sf \binom{M}{S}}=0$. Substituting this into \eqref{eq: communication cost} shows that the achievable communication cost reduces exactly to the optimal expression in Lemma~\ref{lm: optimal communication cost}.

Let us then consider the case where  $\rm N-N_r+2 \le \sf S\leq M$.
We compare the achievable communication cost with the non-secure optimum by examining their ratio:
\begin{equation}
    \sf \frac{R}{\Rsf_{\rm n}^*} = \frac{N_r-(N-M)}{N_r-\beta(N-M)},
\end{equation}
where $\beta=\sf \frac{\binom{N}{S}}{\binom{N}{S}-\binom{M}{S}}$.

From the above expression, $\sf \frac{R}{\Rsf_{\rm n}^*}$ is maximized when $\beta$ is maximized. We therefore analyze $\beta$, which can be rewritten as:
\begin{equation}
    \beta = \sf \frac{1}{1-\frac{M(M-1)\ldots(M-S+1)}{N(N-1)\ldots(N-S+1)}}.
\end{equation}

Notice that $1>\sf \frac{M}{N}>\frac{M-1}{N-1}>...>\frac{M-S+1}{N-S+1}$, it can be verified that $\beta$ attains its maximum when $\sf S=2$. Under the feasibility condition $\sf S\ge N-N_r+2$, this implies $\sf N=N_r$. Substituting these conditions into the ratio yields:
\begin{equation}
    \frac{\Rsf}{\Rsf_{\rm n}^*} \le {\sf \frac{N+M}{N}}\le 2,
\end{equation}
which completes the proof.

\section{Encodability and Decodability Proof}
\label{pr: encodability}
\label{pr: decodability}

In this appendix, we prove that there exists a coding matrix $\Cm$ satisfying both the encodability and decodability constraints introduced in Section~\ref{sec: general scheme}. The proof proceeds in two steps. First, we derive feasibility conditions that ensure sufficient independent vectors are available to construct the required submatrices. Second, we prove existence by showing that the corresponding determinant polynomials are jointly nonzero over a sufficiently large finite field.

Assume each key is partitioned into $\alpha$ pieces, and recall that $\Fm_3={\bf I}_{\alpha \sf \binom{N}{S}}$. Decompose the coding matrix $\Cm$ as $\Cm=[\Cm_G,\Cm_Q]$, where $\Cm_G \in \mathbb{F}_q^{r\Nsf \times n}$ corresponds to gradient rows and is unconstrained, and $\Cm_Q \in \mathbb{F}_q^{r\Nsf \times \alpha \sf \binom{N}{S}}$ corresponds to key rows and contains structured zeros imposed by encodability. We group rows server-wise as $\Cm=[\Cm_1^T,\ldots,\Cm_\Nsf^T]^T$, where each $\Cm_i$ has $r$ rows.
\begin{enumerate}
    \item \textbf{Encodability Constraint:} The matrix $\Cm_Q(\overline{\Dc_i}, \cdot)$ must be full rank. This condition guarantees that the corresponding linear system is solvable, ensuring that the elements of $\Fm_2$ can be properly constructed.
    \item \textbf{Decodability Constraint:} For every $\Uc \subseteq [\Nsf]$ with $|\Uc|=\sf N_r$, the matrix $\Cm_\Uc$ must be full rank. This ensures that, upon receiving transmissions from any $\sf N_r$ servers, the user can recover the desired quantity $\Fm\Wm$ by inverting (or eliminating) $\Cm_\Uc$.
\end{enumerate}

\subsection{Feasibility}
Let $\Sc_n$ denotes the set of keys available to server $n$, and define $\Sc_\Nc=\cup_{n\in \Nc} \Sc_n$ as the set of keys available to at least one server in $\Nc$. Let $\Omega_\Nc=|\Sc_\Nc|$ denote the number of keys available to at least one server in $\Nc$. We first characterize the number of available independent vectors in $\Cm_Q(\Nc, \cdot)$.

\begin{lemma}
\label{lm: dimension}
For any server set $\Nc \subseteq[\Nsf]$, the number of available independent vectors to $\Cm_Q(\Nc,\cdot)$ is $\alpha \Omega_\Nc$.
\end{lemma}

{\it Proof}.
Since $\Fm_3$ is an identity matrix, each key piece corresponds to a distinct orthogonal basis vector. For servers in $\mathcal{N}$, the encodability constraint imposes $\Cm_Q(\mathcal{N}, \cdot) F_3(\cdot, \overline{H_{\mathcal{N}}}) = 0$, where $\overline{H_{\mathcal{N}}}$ indexes key pieces unavailable to all servers in $\mathcal{N}$. Hence, the rows of $\Cm_Q(\mathcal{N}, \cdot)$ must lie in the left nullspace of $\Fm_3(\cdot, \overline{H_{\mathcal{N}}})$. Because $\Fm_3$ is identity, removing columns indexed by $\overline{H_{\mathcal{N}}}$ removes the corresponding basis vectors, and the remaining basis vectors correspond exactly to keys in $S_{\mathcal{N}}$. Each such key contributes $\alpha$ independent directions. Therefore, the available dimension is $\alpha \Omega_{\mathcal{N}}$.\hfill$\square$

We now derive feasibility conditions.

{\it Encodability Feasibility Condition}.
Servers in $\mathcal{N} \subseteq \overline{\Dc_i}$ collectively transmit $r|\mathcal{N}|$ rows. By Lemma~\ref{lm: dimension}, the number of available independent directions equals $\alpha \Omega_{\mathcal{N}}$. Hence, a necessary condition for full rank is
\begin{equation}
\label{eq: encodability feasibility}
    \alpha \Omega_{\mathcal{N}} \ge r|\mathcal{N}|, \quad \forall \mathcal{N} \subseteq \overline{\Dc_i}.
\end{equation}
If this holds for all $\mathcal{N} \subseteq \overline{\Dc_i}, i \in [\Ksf]$, then $\Cm_Q(\overline{\Dc_i}, \cdot)$ can be constructed full rank.

{\it Decodability Feasibility Condition}.
Consider any $\Uc \subseteq [\Nsf]$ with $|\Uc| = \sf N_r$. For $\mathcal{N} \subseteq \Uc$, the total number of available independent directions equals $n + \alpha \Omega_{\mathcal{N}}$, where $n$ comes from the unconstrained gradient part $\Cm_G$. Since servers in $\mathcal{N}$ transmit $r|\mathcal{N}|$ rows, a necessary condition for full rank is
\begin{equation}
    n + \alpha \Omega_{\mathcal{N}} \ge r|\mathcal{N}|, \quad \forall \mathcal{N} \subseteq \Uc.
\end{equation}

Appendix~\ref{pr: feasibility} verifies that both inequalities hold under the communication cost derived in Theorem~\ref{thm: communication cost}. We now prove existence assuming feasibility.

\subsection{Existence}
We first parameterize the free variables. Let $\Cm=[\Cm_G,\Cm_Q]\in \mathbb{F}_q^{r\Nsf\times r{\sf N_r}}$, where $\Cm_G \in \mathbb{F}_q^{r\Nsf\times n}$ is fully free and $\Cm_Q \in \mathbb{F}_q^{r\Nsf\times \alpha \sf \binom{N}{S}}$ has a precise zero pattern. Let $\Theta \subseteq [r\Nsf] \times [\alpha \sf \binom{N}{S}]$ be the set of free positions in $\Cm_Q$, and positions $\Theta^c$ are forces to zero. Let the collection of all free scalar variables be ${\bf x}=\{x_1,\ldots,x_m\}$, where $m=nr\Nsf+|\Theta|$. Every entry of $\Cm$ is either one of the $x_i$ or a constant 0. Hence every minor determinant becomes a multivariate polynomial in $\bf x$. Then we encode each feasibility condition as a polynomial is nonzero.

{\it Encodability condition:} For each $\overline{\Dc_i}$ where $i\in[\Ksf]$, consider the submatrix $\Cm_Q(\overline{\Dc_i}, \cdot)$ with dimension $r|\overline{\Dc_i}| \times \alpha \sf\binom{N}{S}$. The columns corresponding to unavailable keys are all zero, after removing these columns we have the submatrix $\Cm_Q(\overline{\Dc_i}, \Hc_{\overline{\Dc_i}})$ with dimension $r|\overline{\Dc_i}| \times \alpha \Omega_{\overline{\Dc_i}}$. In our scheme we have $r|\overline{\Dc_i}| = \alpha \Omega_{\overline{\Dc_i}}$. Then we define
\begin{equation}
    p_{\overline{\Dc_i}}({\bf x}) = det(\Cm_Q(\overline{\Dc_i}, \Hc_{\overline{\Dc_i}})) \in \mathbb{F}_q[{\bf x}].
\end{equation}

Then $\Cm_Q(\overline{\Dc_i}, \cdot)$ is full rank $\Longleftrightarrow p_{\overline{\Dc_i}}({\bf x})\neq 0$.

{\it Decodability condition:} For each $\Uc \subseteq [\Nsf]$ where $|\Uc|=\sf N_r$, consider the submatrix $\Cm(\Uc, \cdot)$ with dimension $r{\sf N_r} \times r{\sf N_r}$. We define
\begin{equation}
    q_\Uc({\bf x}) = det(\Cm(\Uc, \cdot)) \in \mathbb{F}_q[{\bf x}].
\end{equation}

Then $\Cm(\Uc, \cdot)$ is full rank $\Longleftrightarrow q_\Uc({\bf x})\neq 0$.

We build our main polynomial.
\begin{equation}
    P({\bf x}) = \left(\prod_{i\in[\Ksf]}p_{\overline{\Dc_i}}({\bf x})\right) \left(\prod_{\Uc \subseteq [\Nsf], |\Uc|={\sf N_r}}q_\Uc({\bf x})\right).
\end{equation}

Our goal is to show there exists $x$ such that $P({\bf x}) \neq 0$.

Each $p_{\overline{\Dc_i}}({\bf x})$ is multilinear in the entries of $\Cm_Q(\overline{\Dc_i}, \Hc_{\overline{\Dc_i}})$, hence $deg(p_{\overline{\Dc_i}})\le r|\overline{\Dc_i}|$. Each $q_\Uc({\bf x})$ is multilinear in the entries of $\Cm(\Uc, \cdot)$, hence $deg(q_\Uc)\le r{\sf N_r}$. Thus
\begin{equation}
    deg(P) \le r|\overline{\Dc_i}| \Ksf  + r \sf N_r \binom{N}{N_r}.
\end{equation}

Pick $\bf x$ uniformly at random from $\mathbb{F}_q^m$ (i.i.d. uniform for each free variable). Since $P$ is a nonzero polynomial,
\begin{equation}
    Pr[P({\bf x})=0] \le \frac{deg(P)}{q} \le \frac{r|\overline{\Dc_i}| \Ksf  + r \sf N_r \binom{N}{N_r}}{q}. 
\end{equation}

If $q> r|\overline{\Dc_i}| \Ksf  + r \sf N_r \binom{N}{N_r}$, then $Pr[P({\bf x})\neq 0]>0$, so there exists at least one assignment of the free entries over  that satisfies both properties. This gives a fully rigorous existence proof.

\section{Feasibility Proof}
\label{pr: feasibility}
In this appendix, we verify that the feasibility conditions derived in Appendix~\ref{pr: encodability} are satisfied under the communication cost specified in Theorem~\ref{thm: communication cost}. 

We first compute $\Omega_{\Nc}$.

\begin{lemma}
\label{lm: number of keys}
For any server set $\Nc \subseteq [\Nsf]$,
\begin{equation}
\label{eq: unique owned keys}
\Omega_{\Nc} = \binom{\Nsf}{\Ssf}-\binom{\Nsf-|\Nc|}{\Ssf}=\sum_{k=1}^\Ssf \binom{|\Nc|}{k}\binom{\Nsf-|\Nc|}{\Ssf-k},
\end{equation}
\end{lemma}

\textit{Proof of \eqref{eq: unique owned keys}.}
We focus on the servers in the set $\Nc$. There are $\binom{\Nsf}{\Ssf}$ uncoded groupwise keys in total. A key is unavailable to $\Nc$ if and only if it is shared exclusively among the remaining $\Nsf-|\Nc|$ servers. The number of such keys is therefore $\binom{\Nsf-|\Nc|}{\Ssf}$. Subtracting these unavailable keys from the total yields the first expression \eqref{eq: unique owned keys}.
The second expression follows by classifying keys according to the number of servers shared within the set $\Nc$. For a given $k \in \{1,\ldots,\Ssf\}$, the number of keys shared by exactly $k$ servers in $\Nc$ and $\Ssf-k$ servers in $[\Nsf]\setminus\Nc$ is $\binom{|\Nc|}{k}\binom{\Nsf-|\Nc|}{\Ssf-k}$. Summing over all possible values of $k$ gives the second expression in \eqref{eq: unique owned keys}. \hfill$\square$

\subsection{Encodability Feasibility Proof}

In this part, we prove the following equation
\begin{equation}
\label{eq: loose endodability feasibility}
    \forall \Nc \subseteq [\Nsf], |\Nc| \le {\sf N-M} : \alpha \Omega_\Nc \ge r|\Nc|
\end{equation}
holds in our scheme, which implies the encodability feasibility condition in Equation~\eqref{eq: encodability feasibility}.

In our scheme, each key is partitioned into $\alpha = \sf N-M$ pieces, and each server sends $r=\sf \binom{N}{S}-\binom{M}{S}$ messages. Dividing both sides of~\eqref{eq: loose endodability feasibility} by $|\Nc|$, it suffices to prove
\begin{equation}
    \frac{\alpha \Omega_\Nc}{|\Nc|} \ge r, \forall \Nc \subseteq [\Nsf], |\Nc| \le \sf N-M.
\end{equation}

Since $\frac{r}{\alpha}$ is constant, we analyze the function
\begin{equation}
    \frac{\alpha \Omega_\Nc}{|\Nc|}=\frac{\alpha \left({\sf \binom{N}{S}-\binom{N-|\Nc|}{S}}\right)}{|\Nc|},
\end{equation}
which depends only on $x=|\Nc|$. Define
\begin{equation}
    f(x)=\frac{\alpha \left({\sf \binom{N}{S}-\binom{N-{\it x}}{S}}\right)}{x}.
\end{equation}

The desired inequality becomes
\begin{equation}
    f(x) \ge r, \forall x\in[\sf N-M].
\end{equation}

Notice that $f({\sf N-M})=r$. Therefore, it suffices to show that $f(x)$ is decreasing in $x$. This is equivalent to proving
\begin{equation}
\alpha\frac{\binom{\Nsf}{\Ssf}-\binom{\Nsf-x}{\Ssf}}{x}
>
\alpha\frac{\binom{\Nsf}{\Ssf}-\binom{\Nsf-(x+1)}{\Ssf}}{x+1}.
\end{equation}
Rearranging terms yields
\begin{equation}
\binom{\Nsf}{\Ssf}>(x+1)\binom{\Nsf-x}{\Ssf} - x\binom{\Nsf-x-1}{\Ssf}.
\end{equation}
Using the binomial identity $\binom{n}{k}=\binom{n-1}{k}+\binom{n-1}{k-1}$, the right side can be rewritten as
\begin{equation}
\binom{\Nsf-x-1}{\Ssf} + (x+1)\binom{\Nsf-x-1}{\Ssf-1}.
\end{equation}

To verify the inequality, observe that
\begin{align}
\label{eq: split two parts}
\binom{\Nsf}{\Ssf}
&\overset{\eqref{eq: unique owned keys}}{=}
\binom{\Nsf-x-1}{\Ssf}+\sum_{k=1}^\Ssf \binom{x+1}{k}\binom{\Nsf-x-1}{\Ssf-k} \notag \\
&>\binom{\Nsf-x-1}{\Ssf}+ (x+1)\binom{\Nsf-x-1}{\Ssf-1},
\end{align}
where the strict inequality follows since the summation contains additional positive terms for $k\ge 2$. We finish the proof.

\subsection{Decodability Feasibility Proof}
In this part, we need to prove
\begin{equation}
    n+\alpha \Omega_\Nc \ge r|\Nc|, \forall \Nc \subseteq \Uc, \Uc \subseteq [\Nsf], |\Uc|={\sf N_r}.\notag
\end{equation}

Again let $x=|\Nc|$ and define
\begin{equation}
    g(x) = \frac{n+\alpha \left(\sf \binom{N}{S}-\binom{N-{\it x}}{S}\right)}{x},x\in[\sf N_r].
\end{equation}

The decodability condition becomes $g(x) \ge r$. By construction of the communication cost, we have
\begin{equation}
    g({\sf N_r}) = r.
\end{equation}

Since $f(x)$ decreases in $x$ and $\frac{n}{x}$ also decreases in $x$, $g(x)$ is decreasing in $x$. Hence
\begin{equation}
    g(x) \ge g({N_r}) = r,
\end{equation}
for all $x \le \sf N_r$, which establishes decodability feasibility.

Combining the two parts, both feasibility conditions required in Appendix~\ref{pr: encodability} are satisfied under the proposed communication cost.

\section{Security Proof}
\label{pr: security}
We first present the security proof for the illustrative example, where the objective is to verify that $I(X_1,X_2,X_3;g_1,g_2,g_3|g_1+g_2+g_3)=0$. We partition the matrix $\Wm$ into two submatrices $\Wm=[\Wm_G,\Wm_Q]^T$, where $\Wm_G$ consists of the first $n\Ksf$ rows of $\Wm$ corresponding to the gradient segments, and $\Wm_Q$ comprises the remaining $\sf (N-M)\binom{N}{S}$ rows corresponding to the key segments. Assuming $\Lsf=1$, we proceed as follows:
{\allowdisplaybreaks
\begin{subequations}
\begin{align}
&I(X_1,X_2,X_3;g_1,g_2,g_3|g_1+g_2+g_3) \\
=&H(X_1,X_2,X_3|g_1+g_2+g_3)-H(X_1,X_2,X_3|g_1,g_2,g_3) \\
=&H(X_1,X_2,X_3|g_1+g_2+g_3)-H(\Fm_3\Wm_Q|g_1,g_2,g_3) \\
=&H(X_1,X_2,X_3|g_1+g_2+g_3)-H(\Fm_3\Wm_Q) \label{eg: eq0} \\
=&H(X_1,X_2,X_3|g_1+g_2+g_3)- 1 \label{eg: eq1} \\
= & H(X_1,X_2,X_3,g_1+g_2+g_3)- H(g_1+g_2+g_3)-1\\
= & H(X_1,X_2,X_3)- H(g_1+g_2+g_3)-1 \label{eg: eq2} \\
= & H(X_1,X_2,X_3)-2     \\
\le & H(X_1)+H(X_2)+H(X_3) -2\\
\le & \frac{2}{3} \times 3-2 = 0, \label{eg: eq3}
\end{align}
\end{subequations}
where (\ref{eg: eq0}) follows from the mutual independence of the gradients and the keys; (\ref{eg: eq1}) holds because $\Fm_3\Wm_Q$ contains three independent key symbols, each of length $\frac{\Lsf}{3}$; (\ref{eg: eq2}) relies on the fact that the transmitted messages $(X_1,X_2,X_3)$ are sufficient to recover the sum $g_1+g_2+g_3$; and (\ref{eg: eq3}) accounts for the total transmission size, as each server transmits $r$ messages, each of length $\frac{\Lsf}{3}$.
}

Next, we provide the security proof for the general case, where we must establish $I\left(X_{1},\ldots,X_{\sf N};g_{1},\ldots, g_{\sf K} | \sum_{k\in [\Ksf]}g_k \right) = 0$. Analogous to the example, we obtain:
{\allowdisplaybreaks
\begin{subequations}
\begin{align}
    &I\left(X_{1},\ldots,X_{\sf N};g_{1},\ldots, g_{\sf K} | \sum_{k\in [\Ksf]}g_k \right) \\
    =&H\left(X_{1},\ldots,X_{\sf N}| \sum_{k\in [\Ksf]}g_k \right)-H\left(X_{1},\ldots,X_{\sf N}|g_{1},\ldots, g_{\sf K} \right) \\
    =&H\left(X_{1},\ldots,X_{\sf N}| \sum_{k\in [\Ksf]}g_k \right)-H\left(\Fm_3\Wm_Q|g_{1},\ldots, g_{\sf K} \right) \\
    =&H\left(X_{1},\ldots,X_{\sf N}| \sum_{k\in [\Ksf]}g_k \right)-H\left(\Fm_3\Wm_Q \right) \label{ge: eq0}\\
    =&H\left(X_{1},\ldots,X_{\sf N}| \sum_{k\in [\Ksf]}g_k \right) - \frac{\sf (N-M)\binom{N}{S}}{n} \label{ge: eq1}\\
    \le&H\left(X_{1},\ldots,X_{\sf N}, \sum_{k\in [\Ksf]}g_k \right) - H\left(\sum_{k\in [\Ksf]}g_k \right) \notag \\
    &- \frac{\sf (N-M)\binom{N}{S}}{n}\\
    =&H\left(X_{\Uc(1)},\ldots,X_{\Uc(\sf N_r)}, \sum_{k\in [\Ksf]}g_k \right) - 1 - \frac{\sf (N-M)\binom{N}{S}}{n} \label{ge: eq2}\\
    =&H\left(X_{\Uc(1)},\ldots,X_{\Uc(\sf N_r)} \right) - 1 - \frac{\sf (N-M)\binom{N}{S}}{n} \label{ge: eq3}\\
    \le& \sum_{i\in[\Uc]}H\left(X_{i}\right) - 1 - \frac{\sf (N-M)\binom{N}{S}}{n}\\
    \le& \frac{r}{n}{\sf N_r} - 1 - \frac{\sf (N-M)\binom{N}{S}}{n} = 0, \label{ge: eq4}
\end{align}
\end{subequations}
where (\ref{ge: eq0}) results from the independence of the gradients and the keys; (\ref{ge: eq1}) follows because $\Fm_3\Wm_Q$ comprises $\sf(N-M)\binom{N}{S}$ key pieces, each of length $\frac{\Lsf}{n}$; (\ref{ge: eq2}) is justified by the construction of $\Xc=\Cm\Fm\Wm$, where the decodability property of $\Cm$ ensures that transmissions from any set of servers $\Uc\in \sf \binom{N}{N_r}$ contain all necessary information; (\ref{ge: eq3}) holds because the sum of gradients $\sum_{k\in [\Ksf]}g_k$ is recoverable from $(X_{\Uc(1)},\ldots,X_{\Uc(\sf N_r)})$; and (\ref{ge: eq4}) reflects that each server transmits $r$ messages, each of length $\frac{\Lsf}{n}$. This concludes the proof of security for the proposed scheme.
}

\section{Derivation of Communication Cost}
\label{pr: communication cost}
We derive the communication cost by identifying the conditions under which the linear system is solvable.

Assume a communication cost of $\Rsf=\frac{r}{n}$, where each server transmits $r$ independent messages and each gradient is partitioned into $n$ segments. Each key is further partitioned into $\alpha$ segments. Under this construction, the coding matrix $\Cm$ has dimension $r\Nsf \times r{\sf N_r}$, and the demand matrix $\Fm$ has dimension $\left(n+\alpha\binom{\Nsf}{\Ssf}\right) \times \left(n\Ksf+\alpha\binom{\Nsf}{\Ssf}\right)$. For the product $\Cm\Fm$ to be well defined, the number of columns of $\Cm$ must equal the number of rows of $\Fm$, yielding the constraint
\begin{equation}
\label{eq: communication cost proof 1}
    r{\sf N_r}=n+\alpha\binom{\Nsf}{\Ssf}.
\end{equation}

Next, consider the submatrix $\Cm(\overline{\Dc_i},\Oc)$, which has dimension $(\Nsf-\Msf)r\times \alpha\binom{\Nsf}{\Ssf}$. By Lemma~\ref{lm: dimension}, the dimension of the available subspace for this matrix is $\alpha \Omega_{\overline{\Dc_i}}$, and by Lemma~\ref{lm: number of keys} $\Omega_{\overline{\Dc_i}}={\sf \binom{N}{S}-\binom{M}{S}}$. For the corresponding linear system to be solvable, the available subspace dimension must be as large as the number of columns, leading to the necessary condition
\begin{equation}
\label{eq: communication cost proof 2}
    \alpha\left({\sf \binom{N}{S}-\binom{M}{S}}\right) \ge (\Nsf-\Msf)r.
\end{equation}

Finally, eliminating $\alpha$ using \eqref{eq: communication cost proof 1} and \eqref{eq: communication cost proof 2}, we obtain the following lower bound on the communication cost:
\begin{equation}
    \Rsf = \frac{r}{n} \ge \sf \frac{\binom{N}{S}-\binom{M}{S}}{\left(\binom{N}{S}-\binom{M}{S}\right)N_r-\binom{N}{S}(N-M)}.
\end{equation}

\end{document}

%% file: macros.tex
\long\def\comment#1{}

\newfont{\bbb}{msbm10 scaled 700}

\newfont{\bb}{msbm10 scaled 1100}

\newcommand{\Cm}{{\bf C}}

\newcommand{\Fm}{{\bf F}}

\newcommand{\Wm}{{\bf W}}

\newcommand{\Cc}{{\cal C}}
\newcommand{\Dc}{{\cal D}}

\newcommand{\Gc}{{\cal G}}
\newcommand{\Hc}{{\cal H}}

\newcommand{\Nc}{{\cal N}}
\newcommand{\Oc}{{\cal O}}

\newcommand{\Sc}{{\cal S}}

\newcommand{\Uc}{{\cal U}}

\newcommand{\Vc}{{\cal V}}
\newcommand{\Xc}{{\cal X}}

\newcommand{\Zc}{{\cal Z}}

\newcommand{\Ksf}{{\sf K}}
\newcommand{\Lsf}{{\sf L}}
\newcommand{\Msf}{{\sf M}}
\newcommand{\Nsf}{{\sf N}}

\newcommand{\Rsf}{{\sf R}}
\newcommand{\Ssf}{{\sf S}}
\newcommand{\Tsf}{{\sf T}}

\newcommand{\be}{\begin{equation}}
\newcommand{\ee}{\end{equation}}
\newcommand{\bea}{\begin{eqnarray}}
\newcommand{\eea}{\end{eqnarray}}

